\documentclass[12pt]{article}

\usepackage{setspace}
\usepackage{graphicx,color}
\usepackage{hangcaption}
\usepackage{amssymb}
\usepackage{amsmath}
\usepackage{amsthm}
\usepackage{subfigure}

\author
{ Megan~Owen,
J.~Scott~Provan
\thanks{M. Owen is with the Department of Mathematics at North Carolina State University, Raleigh, NC, 27695.  E-mail:  maowen@ncsu.edu.}
\thanks{J.S.~Provan is with the Department of Statistics and Operations Research at the University of North Carolina, Chapel Hill, NC, 27599.  E-mail:  provan@email.unc.edu.}}

\title{A Fast Algorithm for Computing Geodesic Distances in Tree Space}

\newcommand{\lengthen}[1]{\addtolength{\textheight}{#1}\addtolength{\textheight}{#1}\addtolength{\headsep}{-#1}}
\lengthen{.3in}

\theoremstyle{plain}
\newtheorem{thm}{Theorem}[section]

\newtheorem{lem}[thm]{Lemma}

\theoremstyle{definition}

\newcommand{\R}{\mathbb{R}}
\newcommand{\T}{\mathcal{T}}
\newcommand{\Or}{\mathcal{O}}
\newcommand{\B}{\mathcal{B}}
\newcommand{\A}{\mathcal{A}}
\newcommand{\bs}{\backslash}

\renewcommand{\S}{\Sigma}
\newcommand{\tr}[1]{ \begin{list}{}{\setlength{\leftmargin}{#1em}} \item}
\newcommand{\tn}{ \item}
\newcommand{\tl}{ \end{list}}

\providecommand{\abs}[1]{\lvert#1\rvert}
\providecommand{\norm}[1]{\lVert#1\rVert}

\graphicspath{{./}{Figs/}}

\begin{document}
\maketitle

\begin{abstract}
Comparing and computing distances between phylogenetic trees are important biological problems, especially for models where edge lengths play an important role.  The geodesic distance measure between two phylogenetic trees with edge lengths is the length of the shortest path between them in the continuous tree space introduced by Billera, Holmes, and Vogtmann.  This tree space provides a powerful tool for studying and comparing phylogenetic trees, both in exhibiting a natural distance measure and in providing a Euclidean-like structure for solving optimization problems on trees.  An important open problem is to find a polynomial time algorithm for finding geodesics in tree space.  This paper gives such an algorithm, which starts with a simple initial path and moves through a series of successively shorter paths until the geodesic is attained.
\end{abstract}

\pagenumbering{arabic}
\section{Introduction}

A phylogenetic tree describes the evolutionary history of a set of organisms, with the leaf vertices representing the organisms and the interior vertices representing points at which the evolutionary history branches.  Researchers use different criteria and methods for constructing phylogenetic trees from available data about the set of organisms, which can result in several possible trees or a distribution of trees describing the phylogenetic history.  For example, reconstructing the most likely tree for different genes may yield different trees \cite{RokasEtAl03}; different reconstruction methods can also produce different trees on the same set of organisms \cite{KF94}.  Thus a way is needed to quantitatively compare different phylogenetic trees, by computing some metric describing the differences between them.  Many such distance measures have been proposed, including the Robinson-Foulds or partition distance \cite{RF81}, the Nearest Neighbor Interchange (NNI) distance \cite{R71}, the Subtree-Prune-and-Regraft (SPR) distance \cite{H90}, and the Tree Bisection and Reconnection (TBR) distance \cite{AllenSteel01}.  These measures tend to emphasize the differences in topologies between the trees, and often do not account for edge lengths.  If the edge lengths represent such information as number of mutations between speciation events, we lose important information by ignoring them.  Worst yet, most of these measures cannot be computed efficiently and so are of little use when applied to large trees.

To address this issue, Billera et al. \cite{BHV01} propose the concept of a continuous {\em tree space}, and its associated {\em geodesic distance}\/ metric, as a natural way to embed and compare phylogenetic trees.  This tree space consists of a set of Euclidean regions, called {\em orthants}, one for each tree topology.    Orthants are joined together whenever one tree topology can be made into another by exchanging edges between the trees.  Within an orthant, the coordinates of each point represent the edge lengths for a particular tree with the topology associated with that orthant.  The geodesic between two trees is the unique shortest path connecting the two associated points in this space.  Thus traversing the geodesic corresponds to continuously transforming one tree into the other.  In contrast to previous measures, geodesic distance incorporates in a natural way edge lengths as well as the tree topology.  Furthermore, the uniqueness of the geodesic between any pair of trees and the continuity of the tree space suggest this framework has useful properties with respect to optimization over trees and to formulating statistical measures associated with trees (\cite{Holmes03} and \cite{Holmes05}).  Other versions of tree space with different metrics or no metric have been investigated in phylogenetics contexts (\cite{Hillis}, \cite{Charleston95}, and \cite{Maddison91} for example) and in combinatorial ones (\cite{TrappmannZiegler98} and \cite{RobinsonWhitehouse96}).

Two algorithms have been previously proposed for computing the geodesic distance:  \textsc{GeoMeTree} \cite{KHK08} and \textsc{GeodeMaps} \cite{Owen09}.  Both these algorithms search through an exponential number of candidate paths to find the geodesic, so their run time is exponential in the size of the trees.  Currently there are no known polynomial-time algorithms to find tree space geodesics, although a polynomial time $\sqrt{2}$-approximation of the geodesic distance was given by Amenta et al. \cite{AGPS07}.  Some combinatorial and geometric properties of the space of phylogenetic trees were also presented in \cite{Owen09}.

This paper presents the first polynomial-time method for computing geodesic distances --- and the associated geodesics --- between trees in tree space.  The algorithm uses a different approach from the previous papers, by starting with a simple path between the trees and transforming it into successively shorter paths until the geodesic is obtained.  At each step, the algorithm identifies one new orthant that intersects the geodesic, and transforms the current path so that it passes through this new orthant in an optimal manner.  By restricting consideration to the orthants intersecting the geodesic, the algorithm makes only a polynomial number of path transformations.  Each new orthant is identified by finding a weighted vertex cover in a specially constructed bipartite graph, which also is a polynomial time problem.

Section~\ref{tree space:sect} describes the tree space in which we define the geodesic distance, along with some important geometric and combinatorial properties relevant to finding the geodesic.  Section~\ref{polynomial algorithm:sect} gives the geodesic path algorithm between trees with disjoint edge sets, and establishes its correctness and complexity.  Section~\ref{common edges alg:sect} explains how to efficiently use the geodesic path algorithm when the trees have common edges.  The final section outlines some interesting problems that extend the scope of this algorithm and the associated structures.

\section{Tree Space and Geodesic Distance}\label{tree space:sect}
This section describes the continuous space of phylogenetic trees and the concept of geodesic distance.  For further details, see \cite{BHV01}.  A \emph{phylogenetic $n$-tree}, or just \emph{$n$-tree}, is a tree $T=(X,{\cal E},\S)$, where $X=\{0,1,\ldots,n\}$ is a labeled set of vertices, called {\em leaves}, of degree 1, and ${\cal E}$ is the set of interior (nonleaf) edges, such that each interior vertex of $T$ has degree at least 3.  The leaf 0 is sometimes identified as the {\em root}\/ of $T$, although we do not distinguish it here.  Each interior edge $e$ is given an associated non-negative {\em length}\/ $|e|$, or $\abs{e}_T$ if we want to emphasize the tree $T$ to which $e$ belongs.  For now we do not attach lengths to the leaf edges of $T$, although the relevant properties of tree space apply as well when leaf edge lengths are present.  At the end of Section \ref{polynomial algorithm:sect} we extend our results to trees with leaf edge lengths.  For our purposes it is most convenient to represent the topology of a tree $T$ by its set $\S$ of {\em splits}\/ of the interior edges, where the split $X_e| \overline X_e$ associated with edge $e$ represents the partition of $X$ induced by removing the edge $e$ from $T$.  In order that these splits actually correspond to a tree, they must be {\em compatible}, that is, for every two edges $e$ and $f$, one of the sets $X_e\cap X_f$, $X_e\cap \overline X_f$, $\overline X_e\cap X_f$, or $\overline X_e \cap \overline X_f$ is empty.  A set of $n-2$ compatible splits uniquely determines the topology of an $n$-tree \cite[Theorem 3.1.4]{SS03}.  Because of this correspondence, we henceforth identify edges in two trees if they correspond to the same split.

Two example 5-trees are given in Figure~\ref{tree_examples:fig}.  One can verify that the six given edges are distinct, but that, for example, the edge $e_1$ in $T$ and the edge $e_6$ in tree $T'$ have compatible splits.
\begin{figure}[ht]
\input{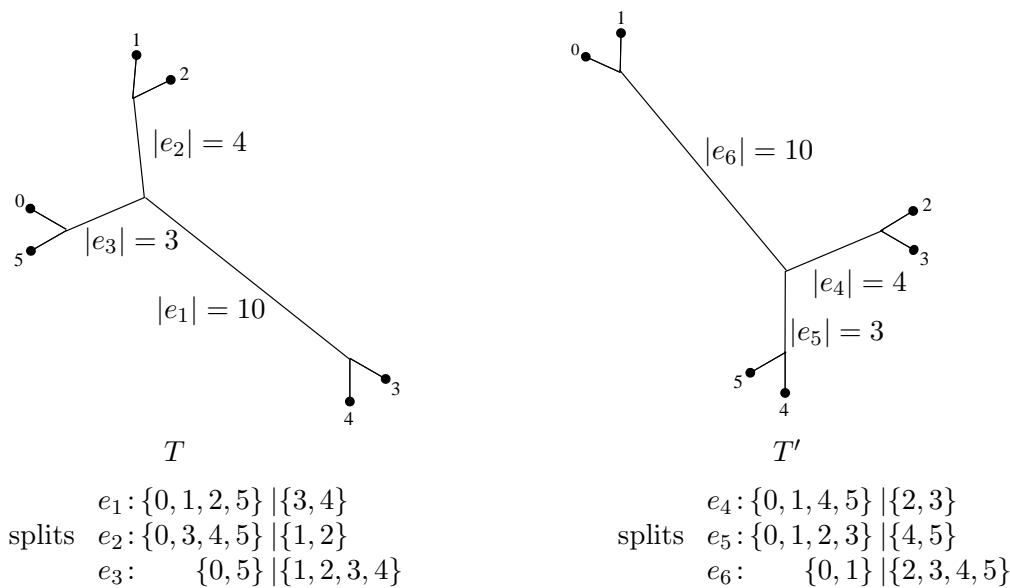}\\[1.5em]
\caption{An example of two 5-trees.}
\label{tree_examples:fig}
\end{figure}

\subsection{Tree Space}
The geometric study of the continuous space of phylogenetic trees $\T_n$, or just {\em tree space}, was pioneered by Billera et al. in \cite{BHV01}.  Fix leaf set $X$ of cardinality $n+1$, where the element labeled 0 is either another leaf or the root.  In $\T_n$ each $n$-tree topology is associated with a unique $k$-dimensional Euclidean orthant (the non-negative part of $\R^k$), where $k$ is the cardinality of the set of edges in that tree topology.  We denote the smallest orthant containing tree $T=(X,{\cal E},\S)$ by $\Or(T)=\Or({\cal E}_+)$, where $T$ is identified in $\Or(T)$ by the vector of lengths of its set ${\cal E}_+$ of positive-length edges.  The interiors of the orthants are disjoint, and represent trees with the same topology but varying (positive) edge lengths.  Thus the maximum-dimension orthants have dimension $n-2$, which is the maximum number of interior edges of an $n$-tree.  Orthants of lower dimension correspond to trees with fewer than $n-2$ edges, and effectively identify the points on the boundary of the higher-dimensional orthants. In particular, we can consider a tree $T$ with $k$ positive-length edges to be on the boundary of any orthant of higher dimension for which some subset of edges in its corresponding tree topology can be contracted  --- equivalently, the length of the edges can be set to zero --- to produce the tree $T$.  Note that as a consequence of this property, each edge in $T$ also appears in the tree topology of every orthant containing $\Or(T)$.

For example, in Figure~\ref{tree_space_geometry}(a), trees $T_1$ and $T_1'$ are represented by distinct points in the same orthant, because they have the same topology but different edge lengths.  Tree $T_2$ is represented in a different orthant;  $T_1$ and $T_2$ have the same edge $e_1$, and so their orthants will be incident.  In particular, the tree $T_3$, with single interior edge $e_1$, can be obtained from $T_1$ ($T_2$ resp.) by setting edges $e_2$ ($e_3$ resp.) to 0 ---  and thus is a point on the $e_1$ axis common to $\Or(T_1)$ and $\Or(T_2)$.

In general, $\T_n$ can be embedded in $\R^N$, where $N=2^n-n-2$ is the number of possible splits on $n + 1$ leaves.  However, as no point in $\T_n$ has a negative coordinate in $\R^{N}$, we often let the positive and negative parts of an axis correspond to different splits.  This can give a more compact representation of the orthants of interest in tree space.  For example, Figure~\ref{tree_space_geometry}(b) illustrates one way the 2-dimensional orthants of five tree topologies in $\T_4$ can be embedded into $\R^3$, by letting $e_3$-$e_5$ and $e_1$-$e_4$ be represented by the the same coordinates.

\begin{figure}[ht]
\input{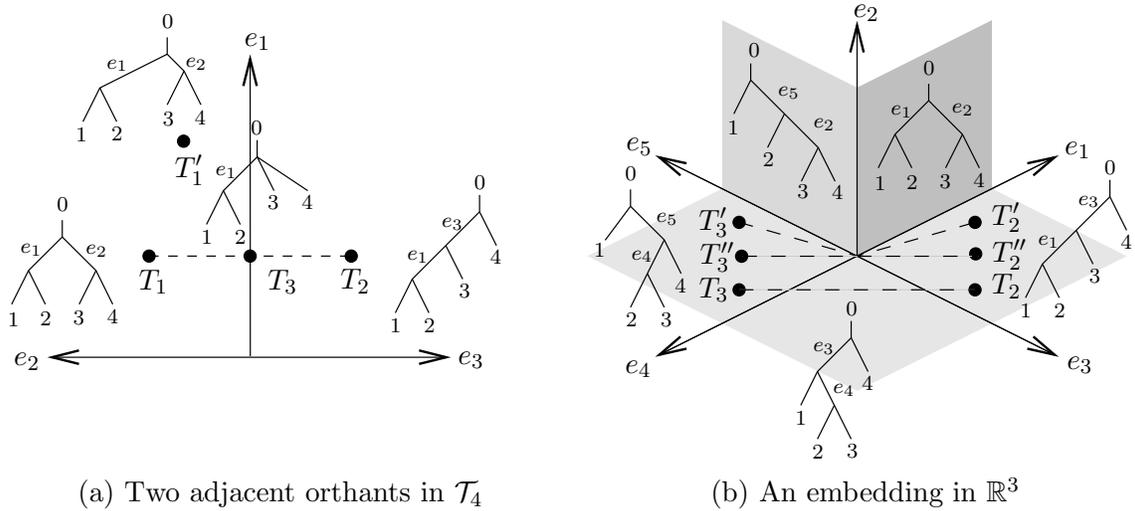}\\[1.5em]
\caption{An example of adjacent orthants in $\T_4$.  Each orthant is labeled with its corresponding tree topology, and the dashed lines indicate the geodesics between the specified trees.}
\label{tree_space_geometry}
\end{figure}

\subsection{Geodesic Distance}
The tree space $\T_n$ has two important properties:
\begin{enumerate}
\item
$\T_n$ is {\em path-connected}, so we can find a parameterized set $\Gamma=\{\gamma(\lambda):\;0\leq\lambda\leq 1\}$ of trees $\gamma(\lambda)\in\T_n$ connecting any two $n$-trees. The simplest such path is the {\em cone path} \cite{BHV01}, which consists of the straight line from the first tree to the origin and the straight line from the origin to the second tree.  We can equivalently think of it as the path formed by contracting all of the edges of each tree at the appropriate constant rates. For any path $\Gamma$, denote its length to be $L(\Gamma)$.  For our purposes $\Gamma$ will always be made up of a sequence of connected line segments, each within its own orthant, and so we can write $L(\Gamma)$ as the sum of the Euclidean lengths of these segments.  This provides a natural metric on $\T_n$ by defining the distance $d(T,T')$ between trees $T$ and $T'$ in $\T_n$ to be the length of a shortest path in $\T_n$ between $T$ and $T'$.
\item
 \cite[Lemma~4.1]{BHV01} $\T_n$ is {\em CAT(0)}, or non-positively curved.  This means, roughly speaking, that triangles in $\T_n$ are ``skinnier'' than the corresponding triangles in Euclidean space.  In particular, let $X$, $Y$, $Z$ be any three points in $\T_n$ and let $W$ be any point on a shortest path from $Y$ to $Z$.  Then if we construct a triangle $xyz$ in Euclidean space with edge lengths $|xy|=d(X,Y)$, $|xz|=d(X,Z)$ and $|yz|=d(Y,Z)$, and let $w$ be the point on $yz$ with $|yw|=d(Y,W)$, then $d(X,W)\leq |xw|$.
\end{enumerate}
As $\T_n$ is CAT(0), there is a {\em unique} shortest path $\Gamma^*$ between any two trees $T$ and $T'$ in $\T_n$.    The path $\Gamma^*$ is called the \emph{geodesic}, and the \emph{geodesic distance} between $T$ and $T'$ is defined as $d(T,T')=L(\Gamma^*)$.  Figure~\ref{tree_space_geometry}(a) gives the geodesic (represented by dotted lines) between two trees in adjacent orthants.  This is clearly the straight line between them.  Figure~\ref{tree_space_geometry}(b) gives three geodesics between trees with no edges in common.  In this case the geodesic is either a cone path (as in the ($T'_2$,$T'_3$)- and ($T''_2$,$T''_3$)-geodesic), or it goes through an intermediate orthant (as in the ($T_2$,$T_3$)-geodesic).  Thus the edge lengths, as well as the tree topology, determine the intermediate orthants traversed by the geodesic.

\begin{figure}[ht]
\centering
\input{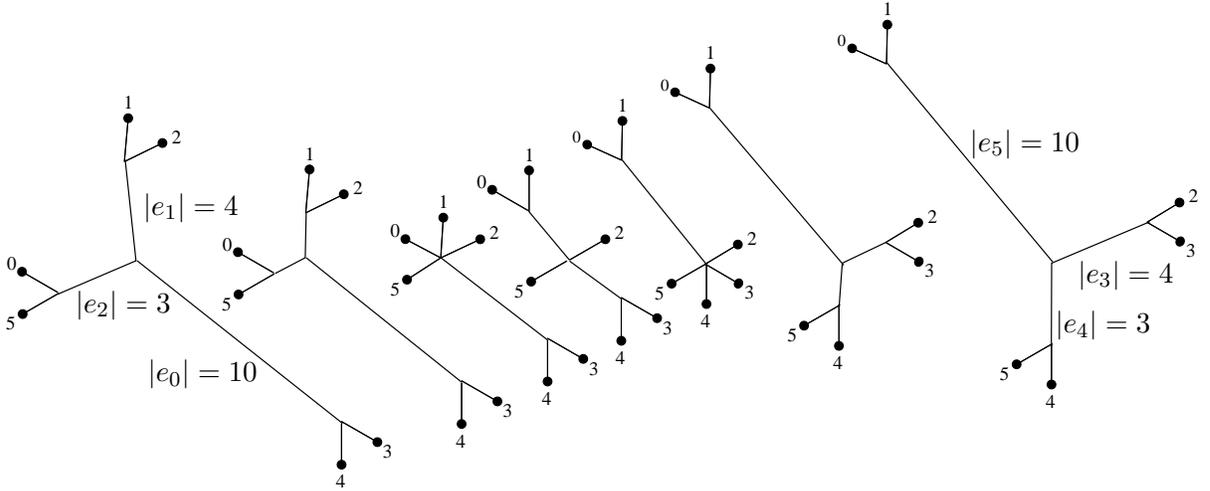}
\caption{A sampling of the trees along the geodesic between the two trees given in Figure~\ref{tree_examples:fig}.}
\label{fig:treepath}
\end{figure}

When the trees have more leaves, the situation becomes more complicated.  For example, the geodesic between the two trees given in Figure~\ref{tree_examples:fig} is illustrated in Figure~\ref{fig:treepath} by a progression of intermediate trees sampled at equidistant points along that geodesic.  This geodesic crosses an orthant boundary at $\lambda=1/3$ and $2/3$, with the intermediate leg corresponding to a tree topology containing only two interior edges.  Two different representations of the sequence of orthants containing the geodesic are given in Figure~\ref{fig:path_space}:  Figure~\ref{fig:path_space}(a) represents the three relevant orthants embedded in $\R^3$, while Figure~\ref{fig:path_space}(b) rotates the three orthants so that the geodesic is represented by a straight line.  Figure~\ref{fig:path_space}(c) gives the tree associated with each leg of the geodesic.  (Points in the figure are labeled with respect to the coordinate system of the orthant containing them.)

\begin{figure}[htb]
\centering
\input{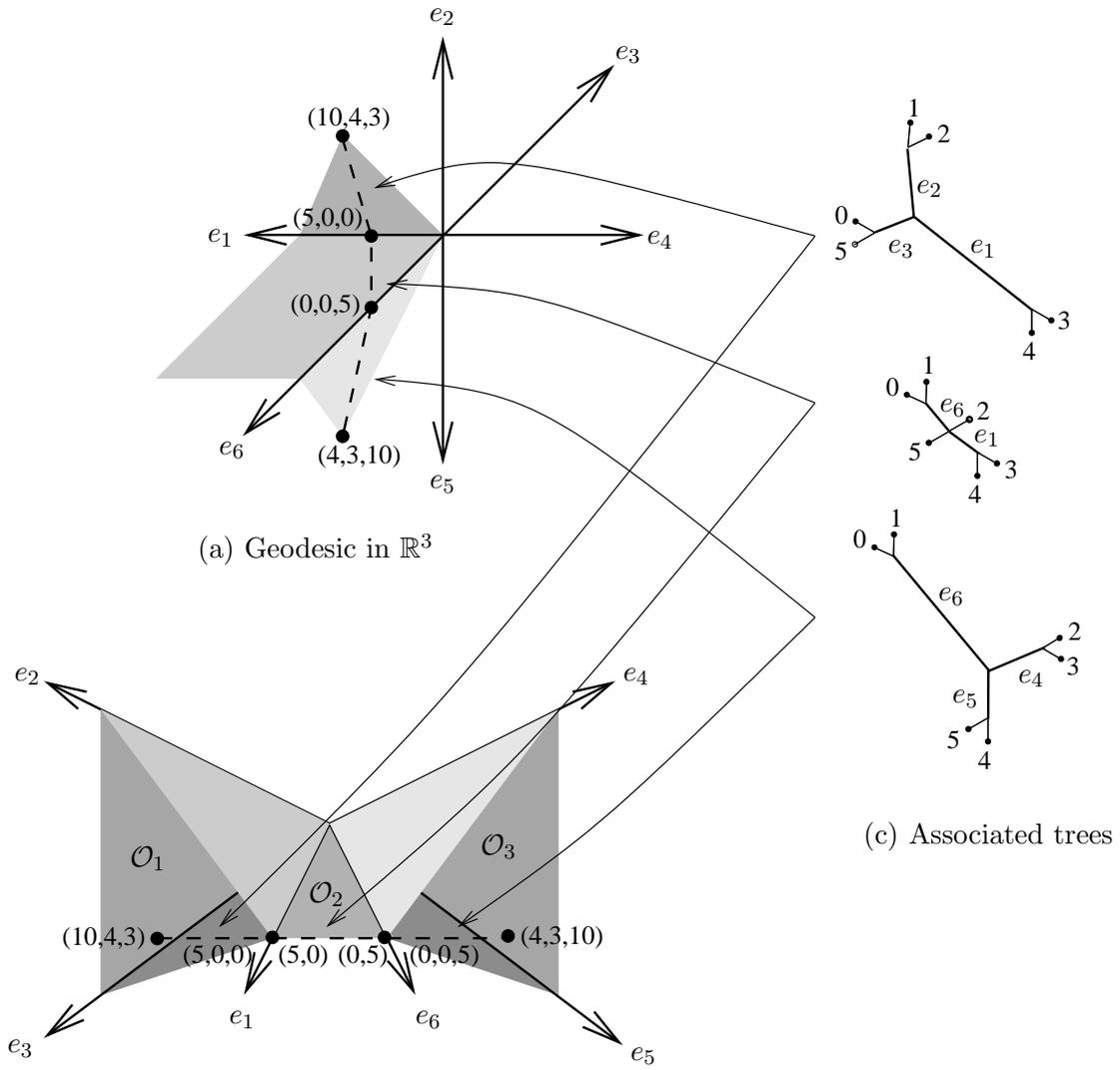}
\centering
\ \\[1em]
\caption{Two representations of the geodesic between the trees in Figure~\ref{tree_examples:fig}.  The geodesic is represented as a dashed line, with the orthants for each topological type in the geodesic identified.}
\label{fig:path_space}
\end{figure}

The definition of geodesic given here differs slightly from the classical definition in a way that is useful to elucidate.  Call a path $\Gamma$ a {\em local geodesic}\/ if there exists some $\varepsilon>0$ so that every subpath of $\Gamma$ of length $\leq \varepsilon$ is the shortest path between its endpoints.  The following result shows that in CAT(0) space this local condition is sufficient to determine the geodesic.  This is also proved in more generality in \cite[Prop. 1.4, Chap. II.1]{BH99}.

\begin{lem}\label{local_optimality}
In a CAT(0) space, every local geodesic is a geodesic.
\end{lem}

\begin{proof}
Let $\Gamma$ be a local geodesic from point $P$ to point $Q$ with associated gauge $\varepsilon$.  Denote by $\Gamma(X,Y)$ the portion of $\Gamma$ between points $X$ and $Y$ on $\Gamma$.  Choose disjoint points $P=P_0,P_1,\ldots,P_k=Q$ on $\Gamma$ such that $L(\Gamma(P_{i-1},P_i))<\varepsilon/2$, $i=1,\ldots,k$.  Then by definition $\Gamma(P_{i-1},P_{i+1})$ is a geodesic for $i=1,\ldots,k-1$.

Now assume by induction that the portion $\Gamma(P_0,P_\ell)$ is a geodesic for $1\leq \ell\leq i$.  Then $L(\Gamma(P_0,P_{i-1}))=d(P_0,P_{i-1})$ by induction, and $L(\Gamma(P_{i-1},P_{i+1}))=d(P_{i-1},P_{i+1})$ by the choice of the $P_i$'s.  Construct the triangle $p_0 p_{i-1} p_{i+1}$ in Euclidean space as specified by the definition of a CAT(0) space, and let $p_i$ be the point on $p_{i-1} p_{i+1}$ with $\abs{ p_{i-1} p_i} = d(P_{i-1},P_i)$.  Then $d(P_0,P_i) \leq \abs{p_0 p_i}$.  But by induction we also have that $\Gamma(P_0,P_i)$ is a geodesic, and so
\begin{align*}
d(P_0,P_i) &=L(\Gamma(P_0,P_i))=L(\Gamma(P_0,P_{i-1}))+L(\Gamma(P_{i-1},P_i)) \\
&= d(P_0,P_{i-1}) + d(P_{i-1},P_i)=\abs{p_0 p_{i-1}} + \abs{ p_{i-1} p_i}.
\end{align*}
This plus the triangle inequality gives $ \abs{ p_0 p_{i-1} } + \abs{ p_{i-1} p_i}=\abs{p_0 p_i}$. But the only way this could happen is if $p_{i-1}$ is on $ p_0 p_i$, which in turn implies that $p_{i-1}$ must also be on $p_0 p_{i+1}$.  Thus $d(P_0 ,P_{i+1})  = d(P_0,P_{i-1}) + d(P_{i-1},P_{i+1})$, and so $\Gamma(P_0,P_{i+1})$ is also a geodesic.  This establishes the inductive step, and the lemma follows.
\end{proof}

It is this result which motivated the idea of this paper.  Namely, we can find a geodesic between trees $T$ and $T'$ by starting with any $(T,T')$-path in $\T_n$, determining whether it is a local geodesic, and if not, transforming it into a shorter $(T,T')$-path.

We define the {\em Geodesic Treepath Problem (GTP)}, to be the problem of finding the geodesic between two trees in $\T_n$.  The remainder of the paper constructs a polynomial-time algorithm for solving GTP.

\subsection{The Path Space of a Geodesic}\label{geometry_of_pathspaces}
Billera et al. \cite{BHV01} showed that the geodesic between $T$ and $T'$ is contained in a sequence of orthants, called a {\em path space}, satisfying certain properties.  These properties were further clarified in \cite{Owen09}.  We  summarize, from \cite[Section 4]{Owen09}, the relevant properties of the shortest path, or geodesic, through a particular path space.  For all path spaces between $T$ and $T'$, the shortest of these path space geodesics will be the geodesic between $T$ and $T'$.

We start with some preliminary assumptions and definitions.  For now we assume that $T$ and $T'$ are \emph{disjoint}, that is, have no common edges.  (We will show at the end of Section  \ref{polynomial algorithm:sect} how to handle common edges between $T$ and $T'$.)  We say that edge sets $A\subset {\cal E}$ and $B\subset {\cal E}'$ are {\em compatible} if every pair of the splits associated with $A$ in $\S$ and $B$ in $\S'$ are compatible, or equivalently, if $A\cup B$ determines a unique $n$-tree.

Let $T=(X,{\cal E},\S)$ and $T'=(X,{\cal E}',\S')$ be disjoint $n$-trees, and let ${\cal A}=(A_1,\ldots,A_k)$  and ${\cal B}=(B_1,\ldots,B_k)$ be partitions of ${\cal E}$  and ${\cal E}'$, respectively, such that the pair $({\cal A},{\cal B})$ satisfies the following property:
\tr 6
[\bf (P1)] \emph{For each $i>j$, $A_i$ and $B_j$ are compatible.}
\tl
Then for all $1 \leq i \leq k$, $B_1 \cup \cdots \cup B_i \cup A_{i+1} \cup \cdots \cup A_k $ is a compatible set, and hence $\Or_i = \Or(B_1 \cup \cdots \cup B_i \cup A_{i+1} \cup \cdots \cup A_k)$ is an orthant in tree space.  Furthermore, the union ${\cal P}=\cup_{i=1}^k \mathcal{O}_i$ of these orthants forms a connected space.  We call ${\cal P}$ a {\em path space}, the pair $({\cal A},{\cal B})$ its {\em support}, and the shortest $(T,T')$-path through ${\cal P}$ the {\em path space geodesic}\/ for ${\cal P}$.

Billera et al. proved the following result \cite[Proposition 4.1]{BHV01} (using the notation $E_i = A_{i+1} \cup \cdots \cup  A_k$ and $F_i = B_1 \cup \cdots \cup  B_i$ for all $1 \leq i \leq k$).
 \begin{thm}\label{path_space_thm}
For disjoint $n$-trees $T$ and $T'$, the geodesic between $T$ and $T'$ is a path space geodesic for some path space between $T$ and $T'$.
\end{thm}
\noindent

In \cite[Section 4]{Owen09}, the requirements for path spaces to contain a geodesic are made more explicit, and the construction of the actual path space geodesic is given.  We summarize the results of this research (Proposition 4.1, Proposition 4.2, Corollary 4.3, Theorem 4.4, and Theorem 4.10 of \cite{Owen09}) below.  For set $A$ of edges, we use the notation $\norm{A}=\sqrt{\parbox{45pt}{$\sum_{e\in A}|e|^2$}}$ to denote the norm of the vector whose components are the lengths of the edges in $A$.

\begin{thm} \label{proper_path}
Let $T=(X,{\cal E},\S)$ and $T'=(X,{\cal E}',\S')$ be two $n$-trees, and let $\Gamma$ be the geodesic in $\T_n$ between $T$ and $T'$.  Then $\Gamma$ can be represented as a path space geodesic with support ${\cal A}=(A_1,\ldots,A_k)$  of ${\cal E}$ and ${\cal B}=(B_1,\ldots,B_k)$ of ${\cal E}'$ which satisfy $(P1)$ plus the following additional property:
\tr 6
[\bf (P2)] $\frac{\norm{A_1}}{\norm{B_1}}\leq \frac{\norm{A_2}}{\norm{B_2}} \leq \ldots \leq \frac{\norm{A_k}}{\norm{B_k}}$.
\tl
\end{thm}
\noindent
We call a path space satisfying conditions (P1) and (P2) a {\em proper path space}, and the associated path space geodesic a {\em proper path}.

The following theorem summarizes results from \cite{Owen09}.

\begin{thm}\label{geodesic_structure}
Let $\Gamma=(\gamma(\lambda):\;0\leq\lambda\leq 1)$ be a proper path between $T$ and $T'$ with support $({\cal A},{\cal B})$.  Then $\Gamma$ can be represented in $\T_n$ with legs
\[
\Gamma^i=\left\{\begin{array}{ll}
\left[\gamma(\lambda):\; \frac{\lambda}{1-\lambda}\leq\frac{||A_1||}{||B_1||}\right],
&i=0\\[.7em]
\left[\gamma(\lambda):\; \frac{||A_i||}{||B_i||}\leq\frac{\lambda}{1-\lambda}\leq\frac{||A_{i+1}||}{||B_{i+1}||}\right],
&i=1,\ldots,k-1,\\[.7em]
\left[\gamma(\lambda):\; \frac{\lambda}{1-\lambda}\geq\frac{||A_k||}{||B_k||}\right],
&i=k\end{array}\right.\]
 where the points on each leg $\Gamma^i$ are associated with tree $T_i=(X,{\cal E}^i,\S^i)$ having edge set
\[\begin{array}{rcl}
{\cal E}^i&=&B_1\cup\ldots\cup B_i\cup A_{i+1}\cup\ldots\cup A_k
\end{array}\]
edge lengths
\[
|e|_{T_i}=\left\{\begin{array}{ll}
        \frac{(1-\lambda)||A_j||-\lambda ||B_j||}{||A_j||}|e|_T&e\in A_j\\[1em]
        \frac{\lambda ||B_j||-(1-\lambda)||A_j||}{||B_j||}|e|_{T'}&e\in B_j\\
       \end{array}\right.
\]
and splits
\[
\begin{array}{rcl}
\S^i_e&=&\left\{\begin{array}{ll}
        X_e|\overline X_e&e\in A_j\\
        X'_e|\overline X'_e&e\in B_j\\
       \end{array}\right.
\end{array}\]
\noindent Furthermore, the length of $\Gamma$ is
\begin{equation}\label{pathlength}
L(\Gamma)=\bigg\Arrowvert(||A_1||,\ldots,||A_k||)+(||B_1||,\ldots,||B_k||)\bigg\Arrowvert.
\end{equation}
\end{thm}
\noindent
{\bf Remark:}  It is easy to see that if any two adjacent support pairs in a proper path space have their ratios in (P2) equal, then combining them again results in a proper path space.  That is, if $(\A, \B)$ is as in Theorem~\ref{proper_path}, and if $\frac{\norm{A_i}}{\norm{B_i}} =  \frac{\norm{A_{i+1}}}{\norm{B_{i+1}}}$ for some $1 \leq i < k$, then $(\A', \B')$, where $\A' = (A_1, ..., A_{i_1}, A_i \cup A_{i+1}, A_{i+2}, ..., A_k)$ and $\B' = (B_1, ..., B_{i_1}, B_i \cup B_{i+1}, B_{i+2}, ..., B_k)$, is also the support of a proper path space.  Further, from the description given in Theorem~\ref{geodesic_structure}, the associated proper path $\Gamma$ does not pass through the interior of the deleted orthant, and hence will also be a proper path for the new path space.  It follows that we can produce a path space for $\Gamma$ for which all of the inequalities in (P2) are {\em strict}. It is shown in \cite[Section 4.2.1]{Owen09} that this in fact is a {\em unique}\/ representation for $\Gamma$.  In this paper, however, we find it more convenient to allow relaxed inequalities in defining proper paths.
\\[1em]
{\bf Example 1:} The cone path between trees $T$ and $T'$ is the path space geodesic for the path space consisting of the two orthants containing the original trees, that is, ${\cal A}=\{{\cal E}\}$ and ${\cal B}=\{{\cal E}'\}$.  This trivially satisfies (P1) and (P2), and the associated proper path is simply the union of the two straight lines connecting $T$ and $T'$ to the origin.
\\[1em]
{\bf Example 2:} For the geodesic given in Figure~\ref{fig:treepath}, the associated path space shown in Figure~\ref{fig:path_space} consists of the starting orthant, the target orthant, and a single intermediate orthant of dimension two on edges $\{e_1,e_6\}$.  Thus the support for this path space will be ${\cal A}=(\{e_2,e_3\},\{e_1\})$ and ${\cal B}=(\{e_6\},\{e_4,e_5\})$, which is proper since
\begin{equation}\label{P2 example}
\frac{||A_1||}{||B_1||}=\frac{||(3,4)||}{10}< \frac{10}{||(3,4)||}=\frac{||A_2||}{||B_2||}.
\end{equation}

  The coordinates (edge lengths) of the path space geodesic as it passes through the intermediate orthant can be ascertained from the representation in Figure~\ref{fig:path_space}(b).  Here the orthants have been positioned so that the geodesic through them is a straight line.  This can be done using the isometric map presented in \cite[Theorem 4.4]{Owen09} from the shaded regions shown in Figure~\ref{fig:path_space}(a) to $\R^2$, and maps the geodesic to the straight line $(\norm{A_1}, \norm{A_2})$ to $(-\norm{B_1}, -\norm{B_2})$.  The length of this line, which is also the length of the path, is
\begin{eqnarray*}
L(\Gamma)&=&\bigg\Arrowvert(||\{e_2,e_3\}||,||\{e_1\}||)+(||\{e_6\}||,||\{e_4,e_5\}||)\bigg\Arrowvert\\[.5em]
&=&\bigg\Arrowvert(||(3,4)||,10)+(10,||(3,4)||)\bigg\Arrowvert=15\sqrt 2
\end{eqnarray*}

Theorem~\ref{proper_path} does not completely characterize the geodesic, in that a path can be proper without being the geodesic.  Consider the two trees $T_2$ and $T_3$ given in Figure~\ref{tree_space_geometry}(b).  The cone path between $T_2$ and $T_3$ is proper, but this path space does not contain the geodesic.  It is necessary to add the  orthant $\Or(\{e_3, e_4\})$ to this cone path space to get the proper path space containing the geodesic.

To check whether we can add such an intermediate orthant to the current candidate path space and shorten the proper path length, we need to check whether we can partition some support pair $(A_i,B_i)$ into two support pairs, such that the addition of the new orthant again results in a proper path space.  That is, we drop some subset of the edges in $A_i$ and add a subset of the edges in $B_i$ to enter the new orthant, and then drop and add the remaining edges to reach the original succeeding orthant.

However, even if such an intermediate orthant exists, adding it to the path space may not result in a shorter proper path.  For example, for the trees $T_2''$ and $T_3''$ in Figure~\ref{tree_space_geometry}(b), we could add orthant $\Or(\{e_3, e_4\})$ to the cone path space and obtain a proper path space, but the proper path for this space will be the same length (actually the same path) as it is for the original path space.  What we need are additional conditions for determining whether adding a specified intermediate orthant will result in a shorter proper path.  As the next result shows, these conditions in fact characterize a proper path as a geodesic.

\begin{thm}{}
\label{lem:when_is_support}
A proper $(T,T')$-path $\Gamma$ with support $({\cal A},{\cal B})$ satisfying {\rm (P1)} and {\rm (P2)} is a geodesic if and only if $({\cal A},{\cal B})$ satisfies the following additional property:
\tr 4
\begin{enumerate}
\item[\bf (P3)] For each support pair $(A_i,B_i)$, there is no nontrivial partition $C_1 \cup C_2$ of $A_i$ and partition $D_1 \cup D_2$ of $B_i$, such that $C_2$ is compatible \vspace{.3em} with $D_1$ and $\frac{ \norm{C_1}}{ \norm{D_1}} < \frac{ \norm{ C_2}}{ \norm{D_2} }$.
\end{enumerate}
\tl
\end{thm}

\begin{proof}
First assume that (P3) does not hold.  Then there exists some support pair $(A_\ell,B_\ell)$, together with partition $C_1 \cup C_2$ of $A_\ell$ and partition $D_1 \cup D_2$ of $B_\ell$, such that $C_2$ is compatible with $D_1$ and $\frac{ \norm{C_1}}{ \norm{D_1}} < \frac{ \norm{ C_2}}{ \norm{D_2} }$.  Let $Y$ be the point where $\Gamma$ passes through the intersection between $\Or_{\ell-1}$ and $\Or_\ell$. Suppose the sequence ratios look like
\[
\cdots\leq\frac{\norm{A_{i-1}}}{\norm{B_{i-1}}} < \frac{\norm{A_i}}{\norm{B_i}} = \cdots =\frac{\norm{A_\ell}}{\norm{B_\ell}} = \ldots = \frac{ \norm{A_j} }{ \norm{B_j} } < \frac{ \norm{A_{j+1}} }{ \norm{B_{j+1}} }\leq\cdots
\]
with the $(i-1)^{st}$ and $(j+1)^{st}$ indices absent if $i=1$ or $j=k$. Then the legs of $\Gamma$ going through $\Or_{i-1}$ and $\Or_j$  have point $Y$ in common and have positive length.  Let $s\neq Y\neq t$ be points  before and after $Y$ on $\Gamma$ in $\Or_{i -1 }$ and $\Or_j$, respectively, and let $\Gamma_{st}$ be the portion of $\Gamma$ between $s$ and $t$, which by choice of $s$ and $t$ consists of two straight lines connected at $Y$.  We will construct a path $\Gamma'_{st}$ from $s$ to $t$ that is shorter than $\Gamma_{st}$, and so $\Gamma$ cannot be a geodesic.

The trees corresponding to $s$ and $t$ have edges $B_1 \cup \cdots \cup  B_{i-1} \cup A_{j + 1} \cup \cdots \cup  A_k$ in common, so by the remarks at the beginning of the subsection, these edges change their length uniformly between $s$ and $t$.  Thus we can restrict attention to the trees $s',Y',t'$ comprised of  $s,Y,t$ with their common edges contracted.  The $s'$ and $t'$ have the edges $\overline A  = A_i \cup \cdots \cup A_j$ and $\overline B = B_i \cup \cdots \cup B_j$, respectively, and by Theorem~\ref{geodesic_structure} the edges in $\overline A $ contract uniformly in $\Gamma$ from $s'$ to $Y'$, and the edges in $\overline B $ expand uniformly in $\Gamma$ from $Y'$ to $t'$.  Thus there is a positive real number $c$ such that the length of any edge $e \in \overline A $ at $s'$ is $c \negthinspace \cdot \negthinspace \abs{e}_{T}$, and a positive real number $d$ such that the length of any edge $f \in \overline B $ at $t'$ is $d \negthinspace \cdot \negthinspace \abs{f}_{T'}$.

Let $\overline C_1 = A_i \cup \cdots \cup A_{\ell-1} \cup C_1$,  $\overline C_2  = C_2 \cup A_{\ell+1} \cup \cdots \cup A_j$,  $\overline D_1  = B_i \cup \cdots \cup B_{\ell-1} \cup D_1$, and  $\overline D_2  = D_2 \cup B_{\ell+1} \cup \cdots \cup B_j$.  Then by the properties of the sets involved, $\frac{ \norm{\overline C_1}}{ \norm{\overline D_1}} < \frac{ \norm{\overline C_2}}{ \norm{\overline D_2} }$.  Let $\overline A',\overline B',\overline C'_1,\overline{C}'_2,\overline D'_1,\overline D'_2$ be the edge sets $\overline{A},\overline B ,\overline C_1,\overline C_2,\overline D_1,\overline D_2$ scaled by $c$ or $d$ to represent the edges at the points $s'$ or $t'$.  Since $\overline C_2$ is compatible with $\overline D_1$, and $\frac{ \norm{ \overline C_1}}{ \norm{ \overline D_1 }} < \frac{ \norm{ \overline C_2 }}{ \norm{ \overline D_2 } }$ implies that $\frac{ c  \norm{ \overline C_1 }}{ d  \norm{ \overline D_1 }} < \frac{ c  \norm{ \overline C_2 }}{ d  \norm{ \overline D_2 } }$, then ${ \cal P'} = ( \{ \overline C_1', \overline C_2'\}, \{ \overline D_1',  \overline D_2'\})$ is a proper path space between $s'$ and $t'$.  Thus Theorem~\ref{geodesic_structure} holds here as well.  Let $\Gamma'_{st}$ be the proper path between $s'$ and $t'$ in ${ \cal P'}$.  Then by Theorem \ref{geodesic_structure}, $\Gamma'_{st}$ passes through the relative interior of the orthants in ${\cal P'}$, but not through $Y'$.  Using Equation (\ref{pathlength}), we get that
\[
L(\Gamma_{st})=||\overline{A}'||+|| \overline{B}'||\mbox{ and } L(\Gamma'_{st})=||(|| \overline{C}'_1||,|| \overline{C}'_2||)+(|| \overline{D}'_1||,||\overline{D}'_2||)||
\]
Now $\frac{ \norm{ \overline{C}'_1}}{ \norm{ \overline{D}'_1}} \neq \frac{ \norm{ \overline{C}'_2}}{ \norm{ \overline{D}'_2} }$  together with the triangle inequality implies that $L(\Gamma'_{st})<L(\Gamma_{st})$, and so $\Gamma$ is not the geodesic between $T$ and $T'$.

Conversely, assume that $\Gamma$ is not a geodesic.  By Lemma~\ref{local_optimality}, this is the case if and only if it is not locally shortest in tree space.  If so, then this must also happen at some bend in $\Gamma$ --- i.e. intersection of orthants --- such that we could cut through some additional orthants to shorten its length.  So suppose $Y$ is such a point, with $\Gamma$ bending at the intersection between $\Or_{i-1}$ and $\Or_i$.

We first consider the simple case where  $\frac{ \norm{A_{i-1}} }{ \norm{B_{i-1}} } <  \frac{ \norm{A_i} }{ \norm{B_i} } < \frac{ \norm{A_{i+1}} }{ \norm{B_{i+1}} }$, with the right or left inequality absent if $i = 1$ or $i = k$.  Let $s\neq Y\neq t$ be points before and after $Y$ on $\Gamma$ in $\Or_{i-1}$ and $\Or_i$, respectively.  Then the section of $\Gamma$ between $s$ and $t$ is not the geodesic from $s$ to $t$.  Let ${\cal P}'$ be the proper path space containing the geodesic from $s$ to $t$, with support $(\cal A', \cal B')$, where ${\cal A'} = (A_1', ..., A_{k'}')$ and ${\cal B'} = (B_1', ..., B_{k'}')$.  By our remark, we can assume that ${\cal P}'$ satisfies (P2) with strict inequalities, and note that $k'$ must be greater than 1.  Let $C_1 = A_1'$, $C_2 = A_i \bs C_1$, $D_1 = B_1'$, and $D_2 = B_i \bs D_1$, and set the length of each edge in $C_1$, $C_2$, $D_1$ and $D_2$ to the length that that edge has in $T$ or $T'$.  Then $(C_1,C_2)$ partitions $A_i$ and $(D_1,D_2)$ partitions $B_i$, and further, $C_2$ and $D_1$ are compatible since ${\cal P}'$ satisfies (P1).

It remains to show that $\frac{ \norm{C_1}}{ \norm{D_1}} < \frac{ \norm{ C_2}}{ \norm{D_2} }$.  By the same argument as above, we have positive constants $c$ and $d$ such that for any edge $e \in A_j$, its length at $s$ is $c \negthinspace \cdot \negthinspace \abs{e}_{T}$, and for any edge $f \in B_j$, its length at $t$ is $d  \cdot  \abs{f}_{T'}$.  Since $(\cal A', \cal B')$ was chosen to satisfy (P2) with strict inequalities, we have $\frac{ c \norm{C_1}}{ d  \norm{D_1}} < \frac{ c  \norm{ C_2}}{ d \norm{D_2} }$, and thus $\frac{ \norm{C_1}}{ \norm{D_1}} < \frac{ \norm{ C_2}}{ \norm{D_2} }$ as desired.

Next we consider the case where  $\frac{\norm{A_{i-1}}}{\norm{B_{i-1}}} < \frac{\norm{A_i}}{\norm{B_i}} = \ldots = \frac{ \norm{A_j} }{ \norm{B_j} } < \frac{ \norm{A_{j+1}} }{ \norm{B_{j+1}} }$, again with the right or left inequality absent if $i = 1$ or $j = k$.  By combining the $A_{i},\ldots,A_j$ and $B_{i},\ldots,B_j$ as per the remark above, we can apply the simple case, so that there exist partitions $C_1 \cup C_2$ of $\overline A = A_i \cup \cdots \cup A_j$ and $D_1 \cup D_2$ of $\overline B  = B_i \cup \cdots \cup B_j$, such that  $C_2$ is compatible with $D_1$ and $\frac{ \norm{C_1}}{ \norm{D_1}} < \frac{ \norm{ C_2}}{ \norm{D_2} }$.  Let $x$ be the common value of $\frac{ \norm{A_\ell}^2 }{ \norm{B_\ell}^2 }$, $\ell=i,\ldots,j$, and let $a_{i\ell}= \norm{A_{\ell} \cap C_i}^2$ and $b_{i\ell}=\norm{B_{\ell} \cap D_i}^2$ for $i=1,2$, $\ell=i,\ldots,j$.  Note that $A_{\ell} \cap C_2$ is compatible with $B_{\ell} \cap D_1$, and thus if (P3) holds, it must be that for all $i\leq\ell\leq j$, $\displaystyle \frac {a_{1\ell}}{b_{1\ell}}\geq \frac{a_{2\ell}}{b_{2\ell}}$, and hence
\[
 \displaystyle\frac {a_{1\ell}}{b_{1\ell}} \geq \frac {a_{1\ell}+a_{2\ell}}{b_{1\ell}+b_{2\ell}}\geq \frac{a_{2\ell}}{b_{2\ell}}.
\]
But
\[
\displaystyle\frac {a_{1\ell}+a_{2\ell}}{b_{1\ell}+b_{2\ell}}= \frac{ \norm{ (A_{\ell} \cap C_1) \cup (A_{\ell} \cap C_2)}^2}{ \norm{ (B_{\ell} \cap D_1) \cup (B_{\ell} \cap D_2)}^2}=\frac{ \norm{A_\ell}^2 }{ \norm{B_\ell}^2 } =x.
\]
Thus we have $\displaystyle\frac {a_{1\ell}}{b_{1\ell}} \geq x \geq \frac{a_{2\ell}}{b_{2\ell}}$ for all $i\leq \ell\leq j$, so that
\[
 \frac{ \norm{C_1}^2}{ \norm{D_1}^2}=\frac{ \norm{ \cup_{\ell = i}^j A_{\ell} \cap C_1 }^2}{ \norm{ \cup_{\ell = i}^j B_{\ell} \cap D_1 }^2}=\frac{\sum_{\ell = i}^j a_{1\ell}}{\sum_{\ell = i}^j b_{1\ell}} \geq x \geq \frac{\sum_{\ell = i}^j a_{2\ell}}{\sum_{\ell = i}^j b_{2\ell}}=\frac{ \norm{ \cup_{\ell = i}^j A_{\ell} \cap C_2 }^2}{ \norm{ \cup_{\ell = i}^j B_{\ell} \cap D_2 }^2}=  \frac{ \norm{C_2}^2}{ \norm{D_2}^2}.
\]
But this contradicts the property of $C_1$, $C_2$, $D_1$ and $D_2$, and thus we have found a partition of an individual pair $(A_{\ell}, B_{\ell})$ also violating (P3), as desired.

\end{proof}
\noindent
{\bf Example 2 (continued):} For the example path given in Figures~\ref{fig:treepath} and \ref{fig:path_space}, if we consider the cone path, then (P3) is violated by sets $(C_1,C_2)=(\{e_2,e_3\},\{e_1\})$ and $(D_1,D_2)=(\{e_6\},\{e_4,e_5\})$, since the inequality in Equation (\ref{P2 example}) is strict.  With the added orthant the resulting proper path becomes the geodesic, since there are no nontrivial partitions of either support pair, and hence (P3) holds.

\section{A Polynomial Algorithm to Solve the Geodesic Treepath Problem}\label{polynomial algorithm:sect}

Theorem~\ref{lem:when_is_support} characterizes when a given proper path $\Gamma$ with support $({\cal A},{\cal B})$ is a geodesic by specifying when a local improvement can be made for the path.  This suggests the following iterative improvement scheme for finding a geodesic between trees $T$ and $T'$:
\begin{enumerate}
\item
Begin with some proper $(T,T')$-path $\Gamma^0$ with support $({\cal A}^0,{\cal B}^0)$.
\item
At each stage we have proper path $\Gamma^\ell$ having support $({\cal A}^\ell,{\cal B}^\ell)$  satisfying condition (P1) and (P2).  Check to see if $({\cal A}^\ell,{\cal B}^\ell)$ also satisfies the condition (P3), and if not, create a new proper path $\Gamma^{\ell+1}$ with support $({\cal A}^{\ell+1},{\cal B}^{\ell+1})$ and having smaller length than $\Gamma^\ell$.
\item
Continue until the geodesic is found.
\end{enumerate}
We now proceed to implement this procedure.  For our starting proper path, choose $\Gamma^0$ to be the {\em cone path}\/ \cite{BHV01}, having support ${\cal A}^0=({\cal E})$ and ${\cal B}^0=({\cal E}')$.  This support vacuously satisfies condition (P1) and (P2), and the path simply corresponds to contracting $T$ and $T'$ uniformly to the origin.

To perform the iterative step, we recast it as a problem on bipartite graphs.  To do this, define the {\em incompatibility graph}\/ $G(A,B)$ between sets $A\subseteq\cal E$ and $B\subseteq\cal E'$ to be the bipartite graph whose vertex set corresponds to $A\cup B$, and whose edges correspond to those pairs $e\in A$ and $f\in B$ such that the corresponding splits $X_e| \overline{ X}_e$ and $X_f| \overline{X}_f$ are incompatible.  An {\em independent set}\/ in $G(A,B)$ is any set of vertices having no edges of $G(A,B)$ between them.  The following lemma follows directly from the definition of compatibility between sets:
\begin{lem}\label{independent_set}
Two edge sets $A\subseteq\cal E$ and $B\subseteq\cal E'$ are compatible if and only if they form an independent set in $G({\cal E},{\cal E}')$.
\end{lem}
We can use Lemma~\ref{independent_set} to restate the problem of determining whether a support $({\cal A}^\ell,{\cal B}^\ell)$ satisfies (P3) as follows:
\tr 4
[\bf Extension Problem]
\tn
[\bf Given:]  Sets $A\subset\cal E$ and $B\subset\cal E'$
\tn
[\bf Question:] Does there exist a partition $C_1 \cup C_2$ of $A$ and a partition $D_1 \cup D_2$ of $B$, such that
\tr 2
[$(i)$] $C_2\cup D_1$ corresponds to an independent set in $G(A,B)$,
\tn
[$(ii)$]$\frac{ \norm{C_1}}{ \norm{D_1}} < \frac{ \norm{ C_2}}{ \norm{D_2} }$\ ?
\tl
\tl
\begin{lem}
A proper path $\Gamma$ with support $({\cal A}^\ell,{\cal B}^\ell)$ is a geodesic if and only if the Extension Problem has no solution for any support pair $(A_i,B_i)$ of $({\cal A}^\ell,{\cal B}^\ell)$.
\end{lem}
The proof follows immediately from Theorem~\ref{lem:when_is_support} and the previous discussion.

We next proceed to solve the Extension Problem.  Since scaling will not affect $(ii)$, we first scale the edge lengths in $A$ and $B$ so that  $\norm{A}=\norm{B}=1$.  By squaring $(ii)$, we get the equivalent condition
\[
\frac{1- \norm{C_2}^2}{ \norm{D_1}^2} < \frac{ \norm{ C_2}^2}{1- \norm{D_1}^2 }
\]
or
\[
\norm{C_2}^2 + \norm{D_1}^2=\displaystyle\sum_{e\in C_2}|e|^2+\sum_{f\in D_1}|f|^2>1.
\]
Thus the Extension Problem reduces to that of finding an independent set in $G(A,B)$ having sufficiently large total weight, where the vertices are weighted by the normalized squares of the edge lengths of $A$ and $B$.

Now note that the pair $C_2$ and $D_1$ form an independent set in $G(A,B)$ if and only if their complements $C_1$ and $D_2$ form a {\em vertex cover}\/ for  $G(A,B)$, that is, every edge of $G(A,B)$ is incident to a vertex of either $C_1$ or $D_2$.  Thus the Extension problem has a solution if and only if the {\em min weight vertex cover}\/ for $G(A,B)$ has weight $\norm{C_1}^2 + \norm{D_2}^2<1$.  (Note that a solution to the extension problem will necessarily result in a nontrivial cover, and hence nontrivial partitions $(C_1,C_2)$ and $(D_1,D_2)$.)
\begin{lem}\label{extension_complexity}
The Extension Problem can be solved in $O(n^3)$ time.
\end{lem}
\begin{proof}
The Min Weight Vertex Cover Problem can be solved on bipartite graphs by a simple extension of the max flow formulation of the Min Cardinality Vertex Cover Problem (see e.g. \cite{ahuja}, Section 12.3), using the vertex weights as capacities on the source and sink arcs.  Maximum flows can be found in $O(n^3)$ time (see e.g. \cite{ahuja}, Section 7.7).
\end{proof}

The solution to the Extension Problem also suggests what the new proper path $\Gamma^{\ell+1}$ should look like.  Namely, if the Extension Problem for support sets $A_i$ and $B_i$ results in a min weight cover by vertex sets $C_1\subset A_i$ and $D_2\subset B_i$ with complements $C_2$ and $D_1$, respectively, then we replace $A_i$ and $B_i$ in $\cal A$ and $\cal B$ by the ordered pairs $(C_1,C_2)$ and $(D_1,D_2)$.  We summarize this in a formal algorithm.

\begin{center}{\bf The GTP Algorithm}\end{center}
\tr 3
[\bf Input:] $n$-trees $T=(X,{\cal E},\S)$ and $T'=(X,{\cal E}',\S')$
\tn
[\bf Output:] The path space geodesic between $T$ and $T'$
\tn
[\bf Algorithm:]\
\tr 2
[\bf Initialize:] Form the incompatibility graph $G({\cal E},{\cal E}')$ between $T$ and $T'$, and set $\Gamma^0$ to be the cone path between $T$ and $T'$ with support ${\cal A}^0=({\cal E})$ and ${\cal B}^0=({\cal E}')$.
\tn
[\bf Iterative step:]
At stage $\ell$, we have proper path $\Gamma^\ell$ with support $({\cal A}^\ell,{\cal B}^\ell)$  satisfying conditions (P1) and (P2).
\tr 2
[{\bf for} each]support pair $(A_i,B_i)$ in $({\cal A}^\ell,{\cal B}^\ell)$, solve the Extension Problem on $(A_i,B_i)$.  Specifically, find a min weight vertex cover for the graph $G(A_i,B_i)$ using vertex weights
\[
w_e=\left\{
\begin{array}{rl}
\frac{|e|^2}{\norm{A_i}^2}&e\in A_i\\
\frac{|e|^2}{\norm{B_i}^2}&e\in B_i
\end{array}
\right.
\]
\tn
[{\bf if} every]min weight cover found above has weight $\geq 1$, then $\Gamma^\ell$ satisfies (P3), and hence is the geodesic between $T$ and $T'$.
\tn
[\bf else]choose any min weight vertex cover $C_1\cup D_2$, $C_1\subset A_i$ and $D_2\subset B_i$ with complements $C_2$ and $D_1$, respectively, having weight $\frac{\norm{C_1}^2}{\norm{A_i}^2} + \frac{\norm{D_2}^2}{\norm{B_i}^2}<1$.  Replace $A_i$ and $B_i$ in ${\cal A}^\ell$ and ${\cal B}^\ell$ by the ordered pairs $(C_1,C_2)$ and $(D_1,D_2)$, respectively, to form new support $({\cal A}^{\ell+1},{\cal B}^{\ell+1})$ with associated proper path $\Gamma^{\ell+1}$.
\tl
\tl
\tl
To establish the correctness of the GTP Algorithm, we need to verify that the resulting path $\Gamma^{\ell+1}$ is indeed proper, i.e. that (P2) holds.

\begin{lem}{}\label{proper_inequality_lemma}
At each stage of the GTP Algorithm, the associated path space satisfies property (P2).
\end{lem}

\begin{proof}
The cone path is trivially proper, so now assume by induction that (P2) holds after stage $\ell-1$ of the algorithm.  Let $({\cal A}^{\ell-1},{\cal B}^{\ell-1})$ be the support for $\Gamma^{\ell-1}$, comprised of the $\ell$ support pairs $(A_1^{(\ell-1)}, B_1^{(\ell-1)}), \ldots,(A_\ell^{(\ell-1)}, B_\ell^{(\ell-1)} )$.  Since the $A_i^{(\ell-1)}$'s are dropped and the $B_i^{(\ell-1)}$'s are added in the order given by their index, we can represent this with visually with a left-to-right ordering of the $A_i^{(\ell-1)}$'s and $B_i^{(\ell-1)}$'s, as in Figure~\ref{proper_inequality:fig}.  Furthermore, since each support pair $(A^{\ell-1}_i,B^{\ell-1}_i)$ must be entirely contained in some support pair $(A^r_j,B^r_j)$ at every stage $r < \ell$ of the algorithm, the added support pairs maintain the existing left-to-right order.   Thus we can depict several partitions with different degrees of refinement in the same diagram for instructional purposes (see Figure~\ref{proper_inequality:fig}).  Since (P1) holds at each stage, then there can be no edges of the incompatibility graph between an element of $A^r_i$ and an element of any group $B^r_j$ to the ``left'' of $A^r_i$.

\begin{figure}[htb]
\centering
\input{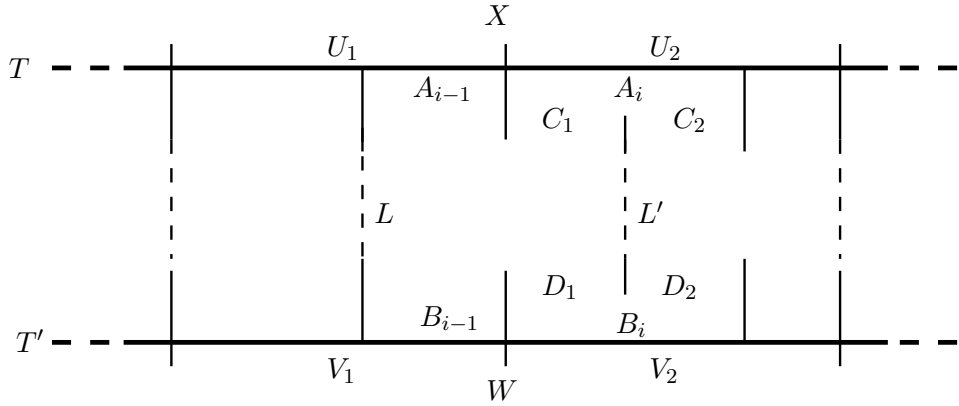}
\caption{The refinements of $\mathcal{E}$ and $\mathcal{E}'$ referred to in the proof of Lemma~\ref{proper_inequality_lemma}.  }
\label{proper_inequality:fig}
\centering
\end{figure}

Now suppose that at stage $\ell$ some support pair $(A_i^{(\ell-1)}, B_i^{(\ell-1)} )$ returns a nontrivial solution to the Extension Problem, comprised of partitions $(C_1,C_2)$ of $A_i^{(\ell-1)}$ and $(D_1,D_2)$ of  $B_i^{(\ell-1)}$, with $C_2$ compatible with $D_1$ and $\frac{ \norm{ C_1}}{ \norm{ D_1}} < \frac{ \norm{C_2}}{ \norm{ D_2 }}$.  Dropping the superscript $(\ell-1)$, we first show that if $i > 1$ then $\frac{ \norm{A_{i-1}} }{ \norm{B_{i-1} } } \leq  \frac{ \norm{ C_1}}{ \norm{D_1}}$.

Let $r < \ell$ be the stage at which sets $A_{i-1}$ and $A_i$ (and hence also sets $B_{i-1}$ and $B_i$) are separated in the partition.  That is, in stage $r-1$ sets $A_{i-1}$ and $A_i$ are in the same partition $X=A^{(r-1)}_j$, and sets $B_{i-1}$ and $B_i$ are in the same partition $W=B^{(r-1)}_j$.  Then in stage $r$ the minimum weight vertex cover $U_1\cup V_2$  is found associated with the Extension Problem on $(X,W)$, creating partitions $(U_1,U_2)$ of $X$ and $(V_1,V_2)$ of $W$, with $A_{i-1}\in U_1$, $A_i\in U_2$, $B_{i-1}\in V_1$, and $B_i\in V_2$.

Consider the vertical lines $L$ and $L'$ in Figure~\ref{proper_inequality:fig}. From  the discussion above, there can be no incompatibility-graph edges from the right of $L$ in $T$ to the left of $L$ in $T'$, and hence $(U_1 \bs A_{i-1}) \cup (V_2 \cup B_{i-1})$ is a vertex cover for $G(X,W)$.  Likewise there can be no incompatibility-graph edges from the right of $L'$ in $T$ to the left of $L'$ in $T'$, and hence $(U_1 \cup C_1) \cup (V_2 \bs D_1)$ is also a vertex cover for $G(X,W)$.  But because $U_1 \cup V_2$ is a minimum weight vertex cover, then it must have weight no greater than either of these covers.  Hence
\begin{eqnarray*}
\frac{ \norm{ U_1}^2 }{ \norm{ U_1 }^2 + \norm{ U_2}^2} + \frac{ \norm{ V_2}^2 }{ \norm{ V_1 }^2 + \norm{ V_2}^2} \leq  \frac{ \norm{ U_1 \bs A_{i-1} }^2 }{ \norm{ U_1 }^2 + \norm{ U_2}^2} + \frac{ \norm{ V_2 \cup B_{i-1} }^2 }{ \norm{ V_1 }^2 + \norm{ V_2}^2} \hspace*{7em}\\[.5em]
=  \frac{ \norm{ U_1 }^2 }{ \norm{ U_1 }^2 + \norm{ U_2}^2}  - \frac{ \norm{ A_{i-1} }^2 }{ \norm{ U_1 }^2 + \norm{ U_2}^2}  + \frac{ \norm{ V_2  }^2 }{ \norm{ V_1 }^2 + \norm{ V_2}^2}  + \frac{ \norm{ B_{i-1}  }^2 }{ \norm{ V_1 }^2 + \norm{ V_2}^2}\\
\end{eqnarray*}
and

\begin{eqnarray*}
\frac{ \norm{ U_1}^2 }{ \norm{ U_1 }^2 + \norm{ U_2}^2} + \frac{ \norm{ V_2}^2 }{ \norm{ V_1 }^2 + \norm{ V_2}^2} \leq  \frac{ \norm{ U_1 \cup C_1  }^2 }{ \norm{ U_1 }^2 + \norm{ U_2}^2} + \frac{ \norm{ V_2 \bs D_1 }^2 }{ \norm{ V_1 }^2 + \norm{ V_2}^2} \hspace{7em}\\[.5em]
= \frac{ \norm{ U_1 }^2 }{ \norm{ U_1 }^2 + \norm{ U_2}^2}  + \frac{ \norm{ C_1  }^2 }{ \norm{ U_1 }^2 + \norm{ U_2}^2}  + \frac{ \norm{ V_2  }^2 }{ \norm{ V_1 }^2 + \norm{ V_2}^2}  - \frac{ \norm{ D_1   }^2 }{\norm{ V_1 }^2 + \norm{ V_2}^2}. \\ \end{eqnarray*}
By cancelling terms and cross-multiplying we get\\
\[
 \frac{ \norm{ A_{i-1} }^2 }{ \norm{B_{i-1}  }^2 }\leq\frac{\norm{ U_1 }^2 + \norm{ U_2}^2}{ \norm{ V_1 }^2 + \norm{ V_2}^2} \leq \frac{ \norm{ C_1 }^2 }{ \norm{ D_1  }^2 }, \\
\]
and the inequality follows.

The argument that if $i<\ell$ then $\frac{ \norm{C_2}}{ \norm{ D_2}} \leq \frac{ \norm{A_{i+1}} }{ \norm{B_{i+1} } }$ is symmetric.  As the other ratios remained unchanged, we have (P2) satisfied after stage $l$ as well, and the lemma follows.
\end{proof}
\noindent {\bf Example 3:}  Lemma~\ref{proper_inequality_lemma} does not necessarily hold outside the context of the GTP algorithm.  In particular, the algorithm may not work correctly if an arbitrary proper path is chosen as the starting path.  Consider tree $T$ in Figure~\ref{tree_examples:fig} and the tree $T'$ given by the three splits $f_1:\{1,3\}|\{0,2,4,5\}$, $f_2:\{1,3,4\}|\{0,2,5\}$, and $f_3:\{1,3,4,5\}|\{0,2\}$ which have lengths $4$, $10$, and $2$, respectively.  Then $\A = (\{e_1, e_2\}, \{e_3\})$ and $\B = (\{f_1, f_2\}, \{f_3\})$ is the support of a proper path between $T$ and $T'$, since
\begin{align*}
\frac{||A_1||}{||B_1||}=\frac{||(10,4)||}{||(4,10)||} = 1< \frac{3}{2}=\frac{||A_2||}{||B_2||}.
\end{align*}
Now (P3) fails for support pair $(\{e_1, e_2\}, \{f_1, f_2\})$, since $f_2$ is compatible with $e_1$ and $\frac{||\{e_2\}||}{||\{f_2\}||} = \frac{4}{10} < \frac{10}{4} = \frac{||\{e_1\}||}{||\{f_1\}||}$.  The refinement $\A' = (\{e_2\},\{e_1\},\{e_3\})$ and $\B' = (\{f_2\},\{f_1\},\{f_3\})$ indicated by the GTP Algorithm, however, is not the support of a proper path, because $\frac{||\{e_1\}||}{||\{f_1\}||} = \frac{10}{4} > \frac{3}{2} = \frac{||\{e_3\}||}{||\{f_3\}||}$.  Instead, the support of our new, shorter proper path is $\A'' = (\{e_2\}, \{e_1,e_3\})$ and $\B'' = (\{f_2\}, \{f_1,f_3\})$, which is not even a refinement of  $(\{e_1, e_2\}, \{e_3\}),(\{f_1, f_2\}, \{f_3\})$.  Note that when we start with the cone path for $\Gamma^0$, however, we obtain the optimal path after a single iteration of the algorithm.

\begin{thm} The GTP Algorithm correctly solves GTP in $O(n^4)$ time.
\end{thm}
\begin{proof}
Lemma~\ref{proper_inequality_lemma} implies that each successful solution to the Extension Problem results in a proper path whose support has one more support pair, so that after at most $n-3$ iterations the algorithm will be unable to find any further nontrivial solutions to the Extension Problem.  It follows that (P3) is satisfied, and so by Theorem~\ref{lem:when_is_support} the resulting path is the geodesic.  Further, we need only solve the Extension Problem on newly created support pairs, since an extension for one support pair will not change the status of any other support pairs.  Thus at most $n-3$ vertex cover problems are be solved throughout the entire algorithm.  The complexity of the algorithm then follows from Lemma~\ref{extension_complexity}.
\end{proof}

Note that in each iteration the new path $\Gamma^{\ell+1}$ satisfies $L(\Gamma^{\ell+1})<L(\Gamma^\ell)$.  This is straightforward to show, although it does not have a direct bearing on the correctness of the algorithm, since the termination of the algorithm is determined only by $\Gamma^\ell$ satisfying property (P3).  It does show, however, that the algorithm is a {\em bona fide}\/ iterative improvement algorithm.
\section{The GTP algorithm with common edges and leaf edge-lengths present}\label{common edges alg:sect}
We finish the analysis by showing how to handle common edges, including the leaf edges, between the terminal trees.  This expanded algorithm allows us to include lengths on leaf edges.  It also allows leaves --- including the root --- to have degree greater than one, since setting the length of the associated leaf edge to 0 contracts the leaf into a vertex with higher degree.

To handle common edges, we first note from \cite[Theorem 2.1]{Owen09} that if $T$ and $T'$ share common edges, then these common edges will be present in every tree on the geodesic, with their lengths changing uniformly between those of their starting and ending trees.   This suggest the following procedure for dealing with common edges:
\begin{enumerate}
\item
Identify the set $C$ of all common nonleaf edges in both trees.  Also let $L$ be the set of leaf edges.
\item
Bisect each edge $e\in C$ by adding midpoint $v_e$.
\item
Separate $T$ and $T'$ at each of the vertices $v_e$ for $e\in C$.  By definition this will leave a collection of pairs of disjoint subtrees $(T(\ell),T'(\ell))$ of $T$ and $T'$, indexed by $\ell=1,\ldots,r$ and with each pair having identical sets of leaves.
\item
For each pair of trees in this collection, apply the GTP Algorithm.   Let $(A_1(\ell),\ldots,A_{k_\ell}(\ell))$ and $(B_1(\ell),\ldots,B_{k_\ell}(\ell))$, $l=1,\ldots,r$, be the support for the associated paths.
\item
The composite path can be described as in Theorem~\ref{geodesic_structure}, with the following modifications:
\begin{enumerate}
\item
For each $\lambda$, the lengths of the edges in the tree $T_i(\ell)$ associated with the pair $(T(\ell),T'(\ell))$ will be as given in Theorem~\ref{geodesic_structure}.  The $\lambda$ is common across all pairs.
\item
Each edge $e\in C$ reconnects the $T(\ell)$'s by reattaching them at $v_e$, with the splits defined accordingly.
\item
The length of each common edge $e\in C\cup L$ on the path is
\[
(1-\lambda)|e|_T+\lambda |e|_{T'}.
\]
\item
The length of $\Gamma$ is\vspace{-1em}
\tr{-4}
\begin{eqnarray*}
L(\Gamma)&=&\bigg\Arrowvert(||A_1(1)||+||B_1(1)||,\ldots,||A_{k_1}(1)||+||B_{k_1}(1)||,\ldots,\\
&&\ \ \ ||A_1(r)||+||B_1(r)||,\ldots,||A_{k_r}(1)||+||B_{k_r}(r)||,\\
&&\ \ \ |e_C|_{_T}-|e_C|_{_{T'}})\bigg\Arrowvert
\end{eqnarray*}
\tl
where $|e_C|_{_{T}}$ and $|e_C|_{_{T'}}$ are the vectors of the lengths of the common edges in the appropriate tree.
\end{enumerate}
\end{enumerate}

Since the partitioning of the tree can be done in linear time, this will not increase the complexity of the algorithm.  An implementation of this algorithm is available at http://www.stat-or.unc.edu/webspace/miscellaneous/provan/treespace.
\section{Conclusion}
This paper presents the first polynomial time algorithm for finding geodesics between phylogenetic trees in tree space, as well as further characterizing properties of geodesics.  This significantly increases the usefulness of the geodesic distance as a modeling tool, since the previous exponential algorithms essentially restricted the geodesic distance measure to trees with fewer than 50 leaves.

We first note that the technique presented here also solves GTP in the case where there is a specific right-left ordering on the non-root leaves of the tree, or equivalently, where the tree must be planar with respect to a given clockwise ordering of the leaves.  An example of such trees are binary search trees.  This condition simply adds to the definition of a tree that the splits of the tree must be {\em noncrossing}, that is, for any split $X_e|\overline X_e$ there are no pairs $v_1,v_2\in X_e$ and $v_3,v_4\in \overline X_e$ which appear in clockwise order $v_1,v_3,v_2,v_4$.  Since if $T$ and $T'$ both satisfy the noncrossing property property, and all of the splits in the intermediate trees on the geodesic between $T$ and $T'$ are made up of the splits of $T$ and $T'$, then these trees must also satisfy the noncrossing property, and so the geodesic for this case is the same as that for the unrestricted case.

The properties and techniques given here potentially apply to the much wider range of problems and measures on trees that make use of the intrinsic Euclidean nature of tree space.  For example, Nye \cite{Nye} compares and groups trees through the idea of ``medial trees'' which serve as representatives for topological types within  data sets of trees.  Hillis et al. \cite{Hillis} also investigate tree sets using distance as the distinguishing feature to find statistical groupings and common features.  Using a more Euclidean-related measure for dissimilarity could allow more powerful statistical techniques to be employed in these situations.

Billera et al. \cite{BHV01} look at the concept of medial trees in their paper by defining the {\em centroid}\/ of a set of points in tree space.  Their definition involves an iterative process that is based on finding a converging sequence of midpoints of geodesics between trees.  The implementation of this would require a fast method of computing geodesics.  Another way of thinking about a centroid in standard Euclidean space, though, is as the point of minimum sum squared distance to the trees.  The framework for finding geodesics here naturally lends itself to finding centroids in this alternate sense as well, and could yield a more direct and efficient way of computing centroids.

A further extension of the idea of centroid comes up in the development of {\em object oriented data analysis (OODA)}\/ as it has been applied to trees \cite{WangMarron07}.  This involves fitting a ``line'' to a set of trees in such a way as to minimize least-squares distances.  The set of nearest points on this line can then be analyzed to yield statistical discriminators that can in turn isolate significant properties of the underlying objects.  This can be done a second time, as a result gaining ``second-order'' information about the set of trees, and so on.  Two problems with OODA have been (a) the difficulty in determining the right concept of ``line'' and ``least-square distances'' when the objects are not Euclidean in nature, and (b) the computational challenge in actually finding these ``least-fit'' objects.  The path space concept presents a compelling model for facilitating both these kinds of analyses, with the CAT(0) property providing the framework for efficient iterative improvement methods to extract useful statistical information in this context.

\section*{Acknowledgment}
This material was based upon work partially supported by the National
Science Foundation under Grant DMS-0635449 to the Statistical and
Applied Mathematical Sciences Institute. Any opinions, findings, and
conclusions or recommendations expressed in this material are those of
the authors and do not necessarily reflect the views of the National
Science Foundation.

\bibliography{polyGeo}


%

\end{document}